\documentclass[11pt]{article}
\usepackage[colorlinks=true,linkcolor=black,citecolor=blue]{hyperref}
\usepackage[utf8]{inputenc}

\usepackage{fullpage}
\usepackage[margin = 2.5cm]{geometry}

\usepackage{amsmath, amssymb, amsthm}

\usepackage{graphicx}
\usepackage{color}
\usepackage{tikz-qtree}  
\usepackage{tcolorbox}
\usepackage{forest}

\usepackage[round]{natbib}

\usepackage{nicefrac}
\usepackage{textcomp}
\usepackage{enumerate}

\newtheorem{remark}{Remark}
\newtheorem{fact}{Fact}
\newtheorem{theorem}{Theorem}[section]

\newtheorem{lemma}[theorem]{Lemma}
\newtheorem{definition}[theorem]{Definition}

\definecolor{Blue}{HTML}{2D2F92}

\usepackage{mdframed}
\newtheorem{result}{Theorem}

\title{Dimension-Accuracy Tradeoffs in Contrastive Embeddings\\ for Triplets, Terminals \& Top-$k$ Nearest Neighbors}

\author{Vaggos Chatziafratis\thanks{Part of this work was done while being supported by a FODSI postdoc fellowship at MIT and Northeastern.}\\University of California, Santa Cruz\\\texttt{vaggos@ucsc.edu} \and Piotr Indyk\\MIT \\\texttt{indyk@mit.edu}}
\date{}

\begin{document}
\maketitle
\thispagestyle{empty}
\begin{abstract}
Metric embeddings traditionally study how to map $n$ items to a target metric space such that distance lengths are not heavily distorted; but what if we only care to preserve the \textit{relative order of the distances (and not their length)}? In this paper, we are motivated by the following basic question: given triplet comparisons of the form ``item $i$ is closer to item $j$ than to item $k$,'' can we find low-dimensional Euclidean representations for the $n$ items that respect those distance comparisons? Such order-preserving embeddings naturally arise in important applications ---recommendations, ranking, crowdsourcing, psychometrics,  nearest-neighbor search--- and  have been studied since the 1950s, under the name of \textit{ordinal} or \textit{non-metric} embeddings. 

Our main results are:

\begin{itemize}
    \item \textbf{Nearly-Tight Bounds on Triplet Dimension:} We introduce the natural concept of \textit{triplet dimension} of a dataset, and surprisingly, we show that in order for an ordinal embedding to be triplet-preserving, its dimension needs to grow as $\frac n2$ in the worst case. This is optimal (up to constant) as $n-1$ dimensions always suffice.
    \item \textbf{Tradeoffs for Dimension vs (Ordinal) Relaxation:} We then relax the requirement that every triplet should be exactly preserved and present almost tight lower bounds for the maximum ratio between distances whose relative order was inverted by the embedding; this ratio is known as (ordinal) \textit{relaxation} in the literature and serves as a counterpart to (metric) distortion.
    \item \textbf{New Bounds on Terminal and Top-$k$-NNs Embeddings:} Going beyond triplets, we then study two well-motivated scenarios where we care about preserving specific sets of distances (not necessarily triplets). The first scenario is \textit{Terminal Ordinal Embeddings} where we want to preserve relative  distance orders to $k$ given items (the ``terminals''), and for that we present matching upper and lower bounds. The second scenario is  \textit{top-$k$-NNs Ordinal Embeddings}, where for each item we want to preserve the relative order of its $k$ nearest neighbors, for which we present lower bounds. 
\end{itemize}

To the best of our knowledge, these are some of the first tradeoffs on triplet-preserving ordinal embeddings and the first study of Terminal and Top-$k$-NNs Ordinal Embeddings.   

\end{abstract}


\newpage
\tableofcontents
\thispagestyle{empty}
\newpage
\section{Introduction}
\label{sec:intro}
\setcounter{page}{1}
Given $n$ items of interest endowed with some abstract notion of ``distance'' (not necessarily a metric), we often wish to represent them as a configuration of $n$ points in some convenient target metric space, commonly a low-dimensional Euclidean space or a tree metric. Having such  representations has been proven crucial for speeding up computation, reducing memory needs and has led to deep algorithmic and mathematical insights. 

Such representations lie at the heart of many applications including nearest-neighbor searching, ad placement, recommendation systems, crowdsourcing, social networks, clustering, visualization and even psychometrics. Since distances encode interesting information about a dataset, the metric embeddings literature has studied methods to preserve those pairwise distances (either exactly or with distortion), and has yielded various tradeoffs  between the faithfulness of the embedding and its dimensionality~(\cite{matouvsek2013lecture}). Our work is motivated by the following two observations: 
\begin{itemize}
    \item First, notice that many of the aforementioned applications \textit{do not rely on the distances per se, but rather they rely on the} \textit{relative order of those distances}. For example, in recommendation systems or online ad placements, the ranking of which $n$ items to show is what matters, whereas pairwise distance lengths are of little importance. Moreover, at the heart of many of these applications is the fundamental problem of nearest-neighbor search~(\cite{andoni2008near})---asking for the closest point from a dataset to a given query point $q$---which is intrinsically a question about relative orderings, rather than absolute distances. 
    \item Second, the $n$ items of interest may lie in an abstract space where even the notion of pairwise distance may be severely underspecified and hard to evaluate. This is especially common in psychometrics~(\cite{torgerson1952multidimensional,thurstone1954measurement,kruskal1964multidimensional,kruskal1964nonmetric}), where humans are asked to answer queries about their feelings,  preferences etc., and in crowdsourcing marketplaces (e.g., Mechanical Turk) where ``workers'' are paid to provide responses to a series of questions. Because humans are surprisingly bad and inconsistent at answering \textit{cardinal} questions (how \textit{much} did you like this movie or this restaurant?), yet very fast and accurate at answering \textit{ordinal} questions (did you enjoy $A$ \textit{more} than $B$?), the deployed queries are usually (paired) comparisons~(\cite{thurstone1954measurement}) between items, and as such they only provide an indirect access to some underlying notion of ``distance'' for the items.
\end{itemize}



\noindent\paragraph{Computational Task.} The natural question that arises in the scenarios described above is whether there exist low-dimensional Euclidean representations that preserve the \textit{relative ordering} of distances (first observation), given perhaps \textit{incomplete} information about how underlying distances are related (second observation). More concretely, we are interested in the following basic question:

\begin{center}
    \textit{Given a collection of (triplet) comparisons of the form ``A is more similar to B than to C'', are there low-dimensional embeddings that respect the relative order of distances (not their length)?}
\end{center}


\noindent\paragraph{Our Contribution.} Our main contribution is to derive several tradeoffs that arise between the faithfulness of the order-preserving embedding and its dimensionality. For the case where the embedding needs to respect all of the (triplet) comparisons (Sec.~\ref{subsec:exact}), we give nearly-tight bounds  for the dimension. For the case where we allow some distances to be inverted (Sec.~\ref{subsec:non-exact}), we present almost tight lower bounds for the ``ordinal relaxation'' of the embedding (this is the analogue of the notion of ``distortion'' in metric embeddings, see definition below). Finally, we also study two new settings motivated by terminal embeddings (Sec.~\ref{sec:terminal}) and by top-$k$-NNs preservation (Sec.~\ref{sec:top-k}).

\subsection{Related Work}

The aforementioned order-preserving embeddings are usually referred to in the literature as \textit{ordinal}, or \textit{non-metric} embeddings, or \textit{monotone} maps. In this paper, the terms ``contrastive embeddings'' and ``ordinal embeddings'' are used interchangeably to refer to the goal of finding embeddings that preserve the relative order of distances (not their exact length); having access to such contrastive information is popular in contrastive learning~(\cite{smith2005contrastive,saunshi2022understanding}). 

We give here the definition: 

\begin{definition} [Ordinal Embedding, Ordinal Dimension]\label{def:ordinal}
Let $X = ([n], \delta)$ be any metric space on $n$ points, and let $\lVert\cdot\rVert$ be a norm on $\mathbf{R}^d$. We say that $\phi: X\to (\mathbf{R}^d,\lVert\cdot\rVert)$ is an \emph{ordinal embedding} if
for every $x, y, z, w \in X$, we have the following: $\delta(x, y) < \delta(w, z) \iff \lVert\phi(x) - \phi(y)\rVert < \lVert\phi(w) - \phi(z)\rVert$. Moreover, the \emph{ordinal dimension} of $X$ is defined to be the smallest
dimension of a Euclidean space into which $X$ can be ordinally embedded.
\end{definition}

The study on ordinal embeddings dates back to the early days of Multi-Dimensional Scaling (MDS) with some of the seminal works of the field in the 1950s and 1960s. At the time, MDS was heavily used in (mathematical) psychology and psychometrics, where the works of~\cite{torgerson1952multidimensional}, \cite{shepard1962analysis,shepard1974representation,cunningham1974monotone}, \cite{kruskal1964multidimensional,kruskal1964nonmetric} laid the foundations of many concepts in the field; these works attempted to formalize explicitly, the implicit connections between human-perceived similarities and differences among items (e.g., faces, tastes, odors etc.) via suitable data representations, and to propose methods for finding such representations obeying subjective descriptions. 

In this early context from psychometrics, our basic question above for triplet comparisons asking ``is item $A$ more similar to $B$ or to $C$'' would be equivalent to asking whether ``$\delta(A,B) < \delta(A,C)$'', which is an ordinal question that is easy for humans to answer (rather than exactly specifying values for the length of the distances). Later works made important steps towards understanding tradeoffs in ordinal embeddings, and the three most related to our work is the paper by~\cite{bilu2005monotone} and by~\cite{alon2008ordinal,buadoiu2008ordinal}, who considered several special cases, where we have access to the \textit{complete} set of all $\binom{\binom{n}{2}}{2}$ distance comparisons (i.e., $\delta(i,j)\lessgtr \delta(k,l), \forall i,j,k,l\in X$)~(\cite{bilu2005monotone,alon2008ordinal}), or  the case when the source and target metrics come from simple families: in~\cite{buadoiu2008ordinal} they provide approximation algorithms to embed unweighted graphs into a line metric and into a tree metric, and they also consider the embedding of unweighted trees into the line. For the case of embedding onto a line or a hierarchical tree (ultrametric), where the triplet comparisons are seen as a set of constraints with the goal of maximizing the number of constraints satisfied by the embedding, there are tight hardness of approximation results~(\cite{guruswami2008beating,chatziafratis2022phylogenetic}). Moreover, several statistical properties and sample complexity bounds based on queries about distances on four points $i,j,k,l$ have been studied in multiple works~\cite{agarwal2007generalized,terada2014local,kleindessner2014uniqueness,ghoshdastidar2019foundations,ghosh2019landmark}.

An important notion for ordinal embeddings  introduced by~\cite{alon2008ordinal} is the notion of \textit{ordinal relaxation}, which serves as the counterpart of the well-studied notion of distortion of metric embeddings (see e.g.,~\cite{indyk20178}):

\begin{definition} [Ordinal Relaxation] Given an ordinal embedding $\phi: (X,\delta)\to (Y,\delta')$, we say it has multiplicative (ordinal) relaxation $\alpha\ge1$, if $\alpha \delta(i,j)<\delta(k,l)\implies \delta'(i,j)<\delta'(k,l)$. 
\end{definition}

To put it simply, significantly different distances (in the original space) should have their relative order preserved by the embedding,\footnote{As pointed out in the original work~\cite{alon2008ordinal}, in ordinal embeddings we want to respect distance equality, but in an ordinal embedding with relaxation $1$, we may break ties.} or equivalently, relaxation is the maximum ratio between two distances whose relative order got inverted by the embedding. Minimum-relaxation ordinal embeddings were originally studied in~\cite{alon2008ordinal} and~\cite{buadoiu2008ordinal}, who established that ordinal embeddings have important differences from metric embeddings. They developed several approximation algorithms for ordinal embeddings on a line ($1$-dimensional Euclidean space) or on a tree, for various interesting cases such as source metrics induced by ultrametrics, or by the shortest-path metric of unweighted trees and unweighted graphs. Notice that by definition, for any source and target metrics, the optimal relaxation is at most the optimal distortion~(\cite{alon2008ordinal}). Moreover, for any $n$-point metric space, observe that a general $O(\log n)$ upper bound on relaxation into $O(\log n)$-dimensional Euclidean space follows easily by Bourgain's metric embedding theorem~(\cite{bourgain1985lipschitz}) coupled with the Johnson-Lindenstrauss lemma~(\cite{johnson1986extensions}). Another well-known fact is that $n$-dimensional Euclidean space suffices to ordinally embed any $X$, i.e., the ordinal dimension for any metric space $X$ on $n$ points is at most $n$ (to be exact, $n-1$ dimensions always suffice)~(\cite{bilu2005monotone,alon2008ordinal}).

\subsection{Motivating Questions}
Despite their long history and the abovementioned general results, ordinal embeddings ---both exact, or with relaxation--- are not well-understood. It is easy to see that if we care about exact preservation of the distance ordering, then the Johnson-Lindenstrauss lemma is of no use in this case: the error parameter $\epsilon$ would have to be tiny as it needs to scale with the smallest distance gap (which could be $O\left(\tfrac1n\right)$), essentially yielding no dimension reduction at all (recall, the final dimension would be $O(\tfrac{\log n}{\epsilon^2})$). Given this failure, we set out to address several basic questions on low-dimensional ordinal embeddings:
\begin{itemize}
    \item \textbf{Q1} (Triplets, Exact \& Relaxed): In many applications, we do not have access to the complete set of $\binom{\binom{n}{2}}{2}$ comparisons for all distances $\delta(i,j)\lessgtr\delta(k,l)$ (as assumed in~\cite{bilu2005monotone,alon2008ordinal,buadoiu2008ordinal}), rather we are given information on \textit{triplets} of items of the form ``item $i$ is closer to item $j$ than to item $k$'' or equivalently, $\delta(i,j)<\delta(i,k)$ for some abstract distance $\delta$. Are there Euclidean representations with low dimension that respect the relative ordering on triplets, either exactly or with relaxation?
    \item \textbf{Q2} ($k$ Terminals, Upper \& Lower Bound): 
    Given a set of $k$ special items $T=\{t_1,\ldots, t_k\}$ (the ``terminals''), we want to find Euclidean representations so as to preserve all distance comparisons from the perspective of each terminal to the rest of the items. Are there algorithms that find such representations? What is the minimum dimension needed?
    \item \textbf{Q3} (Top-$k$-NNs, With \& Without Mixed Comparisons): What if we only care to preserve the relative distances from each item $i$ to its set $\mathrm{NN}(i)$ containing the $k$ nearest neighbors in the dataset, i.e., comparisons of the form $\delta(i,i'), \delta(i,i'')$ for $i',i''\in \mathrm{NN}(i)$? Is the minimum dimension needed different, if we also cared to preserve the distance ordering for mixed comparisons $\delta(i,i')$ vs $\delta(j,j')$ for $i'\in \mathrm{NN}(i), j' \in \mathrm{NN}(j), i\neq j$?
    \item \textbf{Q4} (Different Regimes for $k$): Is the answer to the above questions different for various regimes of the parameter $k$? Is there a difference in the upper or lower bounds when $k=\Theta(n)$, or $k=o(n)$, or even constant independent of $n$ (number of items)?
\end{itemize}
 We believe that such questions are natural steps towards a better understanding of ordinal embeddings; they also are well-motivated from practical considerations, as we discuss in Sec.~\ref{subsec:further}.

\subsection{Our Results} 

We provide several results for each of the questions listed above. Our work was primarily inspired by the lack of theoretical bounds for the basic problem of preserving triplets either exactly or approximately (Q1), and also by the lack of results for terminal and top-$k$-NNs ordinal embeddings (Q2 and Q3), in stark contrast to their counterparts in metric embeddings. Our main results can be summarized informally as follows:
\begin{itemize}
    \item (R1) For preserving the order of all triplets exactly, we introduce the concept of \textit{triplet dimension}, i.e., the minimum dimension needed by an ordinal embedding that respects all triplet comparisons. This is the natural analogue of ordinal dimension from Definition~\ref{def:ordinal} but specialized for triplets (see also~\cite{reiterman1989geometrical} for other geometric notions of dimension). Perhaps surprisingly, we prove that the triplet dimension can grow linearly in $n$ and may need to be at least $\tfrac n2$ (Theorem~\ref{th:exact_triplet}); this is tight because it is a folklore result that $n-1$ dimensions always suffice. In the case where we allow for some of  the triplet orderings to be inverted, we prove a lower bound on the relaxation of $\Omega\left(\tfrac{\log n}{\log d+\log\log n}\right)$ for any dimension $d$ and $n$ (Theorem~\ref{th:relax_triplet}). This is nearly tight (up to $\log\log n$ factors) because using Bourgain's theorem and the JL lemma, we could obtain relaxation $O(\log n)$ for any metric into $O(\log n)$-dimensional Euclidean space.
    \item (R2) For $k$ terminals ($k$ not necessarily fixed), we present a simple, yet optimal algorithm that allows us to embed the dataset into $k$-dimensional space such that for each of the $k$ terminals, all relative distance orders to the remaining $n-k$ points are exactly preserved (Theorem~\ref{th:upper_terminal}). We complement this with tight lower bounds (Theorems~\ref{th:exact_terminal} and~\ref{th:lower_sub}) showing that $\Omega(k)$ dimensions are indeed necessary.
    \item (R3) For preserving the top-$k$-NNs of each point, we present an $\Omega(k)$ lower bound on the dimension (Theorem~\ref{th:topNNs}). As we will see, this is only slightly affected by whether or not we preserve the distance ordering for mixed comparisons $\delta(i,i')$ vs $\delta(j,j')$ for $i'\in \mathrm{NN}(i), j' \in \mathrm{NN}(j), i\neq j$.
    \item (R4) We show how the value of $k$ actually affects the dimension needed for ordinal embeddings. Specifically, there is a difference for the lower bounds we can obtain for $k=\Theta(n)$ and for $k=o(n)$.
\end{itemize}
To the best of our knowledge, these are some of the first results for triplet, terminal, and top-$k$-NNs ordinal embeddings. In addition, our results are often tight, they extend prior works~(\cite{bilu2005monotone,alon2008ordinal,buadoiu2008ordinal}) and complement many of the empirical works for triplets~(\cite{schultz2003learning,tamuz2011adaptively,jamieson2011low,jain2016finite,kleindessner2017kernel,ghosh2019landmark}) or other types of ordinal embeddings~(\cite{agarwal2007generalized,terada2014local,kleindessner2014uniqueness,ghoshdastidar2019foundations}).


\subsection{Our Techniques}
We build upon and extend the tools used in~\cite{bilu2005monotone,alon2008ordinal,buadoiu2008ordinal} who handled the complete case where all $\binom{\binom{n}{2}}{2}$ distance comparisons were available. Here, in order to handle relaxed ordinal embeddings for non-complete inputs, we have to come up with more elaborate constructions (see Theorem~\ref{th:relax_triplet}). Our results use previous constructions of dense high-girth graphs~(\cite{erdregukre,sauer1970existence}) and rely on sampling edges from those graphs, and on several counting arguments depending on which case we deal with. 

More specifically, we present three types of results: lower bounds for triplets with relaxation, lower bounds for various types of exact ordinal embeddings (triplet, terminal, top-$k$-NNs), and (optimal) upper bounds for $k$ terminal embeddings. 

The intuition behind our proofs for getting the relaxation lower bounds for triplets relies on finding a large number of metric spaces with significantly different behavior on triplet distances, at least from the perspective of a single vertex $v$. To find such metric spaces, we start from a high girth graph $G$ (unweighted) and consider an appropriate number of subgraphs $G'$, chosen at random. Notice that if an edge $(u,v)$ was present in $G$, but after sampling, it is not included in $G'$, then the distance $u,v$ in $G'$ suddenly becomes large (as much as the girth). Our derived tradeoffs for ordinal embeddings essentially stem from the known constructions of high girth graphs and any improvement on the latter yields improvement on our lower bounds. 

Regarding exact ordinal embeddings, at their heart, our proofs are based on various counting arguments. We compare the total number of different orderings for the distances defined on triples of points against the total number of ways an ordinal embedding can embed the $n$ points in $d$ dimensions such that a given ordering on all triplets is satisfied. Similar arguments go through for the other types of embeddings we consider. We then observe that when the dimension $d$ is ``small'' compared to $n$, there will be significantly more triplet orderings than distinct embeddings, which will give us the lower bound. A similar approach was carried out in~\cite{bilu2005monotone} for complete instances where all distance comparisons $\delta(i,j)\lessgtr\delta(k,l)$ were given as input.

Finally, for terminal ordinal embeddings we give a direct construction on where to embed the $k$ terminals and the rest of the points such that distance orders from the rest of the points to the terminals are exactly preserved. We use $k$ dimensions and this matches our lower bound for terminal ordinal embeddings.



\subsection{Further Motivation and Related Work}\label{subsec:further}
The notion of ordinal relaxation defined previously was introduced in~\cite{alon2008ordinal}. Ordinal embeddings for structured metric spaces such as tree metrics and ultrametrics were also studied in the work of~\cite{alon2008ordinal} and our original motivation for studying ordinal embeddings came from studying tree metrics and ultrametrics, in particular in the context of triplet-based hierarchical clustering~\citep{chatziafratis2022phylogenetic}. We should note here that the goal in~\cite{alon2008ordinal} was to design approximation algorithms to approximately minimize the ordinal relaxation, in contrast to recent global objective functions used in hierarchical clustering~\citep{dasgupta2016cost,charikar2017approximate,moseley2023approximation,cohen2019hierarchical,charikar2019hierarchicalsoda,charikar2019hierarchical,chatziafratis2020bisect,alon2020hierarchical,naumov2021objective}. 

Regarding triplet-preserving ordinal embeddings (Q1), obtaining information about a dataset based on triplet comparisons like ``which of $j$ and $k$ is closer to $i$'' is used in crowdsourcing and online platforms, to elicit user preferences and to perform downstream tasks such as clustering and nearest neighbor search. Important works, both theoretical and empirical, that focus on triplet embeddings under various settings include ~\cite{schultz2003learning,tamuz2011adaptively,jamieson2011low,van2012stochastic}, and later~\cite{jain2016finite,kleindessner2017kernel,korlakai2016crowdsourced,lohaus2019uncertainty,vankadara2019insights,ghosh2019landmark,fan2020learning, haghiri2020estimation}. Triplet feedback of the form ``$i$ is closer to $j$ than to $k$''  is known to be much more reliable than absolute comparisons and much easier for humans to answer. For example, the question ``are cats similar to tigers?'' might yield conflicting answers, but the triplet query ``are cats more similar to tigers or to dolphins?'' is easier as humans have to pick the ``odd-one-out'' among the $3$ alternatives; in addition, such triplet queries are also useful in the context of computational biology/phylogenetics for inferring ancestry relations in hierarchical clustering~(\cite{byrka2010new,vikram2016interactive,chatziafratis2018hierarchical,emamjomeh2018adaptive,chatziafratis2021maximizing,chatziafratis2022phylogenetic}), and in clustering via hyperbolic embeddings~(\cite{monath2019gradient,chami2020trees}). Triplet feedback is also widely used in contrastive learning (see~\cite{alon2023optimal}).

Regarding ordinal embeddings for terminals (Q2), these are useful in a scenario where we only care to preserve orders for a few $k$ points to the rest of the $(n-k)$ items; this could arise for example, in a facility location or networking application, whenever we have a network comprising many clients and only $k$ servers, and we want to have a simple data structure preserving the client-to-server service times (captured by the distance), but we do not care about client-to-client distance preservation. This is the analogous notion to metric embeddings with terminals that was introduced by~\cite{elkin2017terminal} and later studied in~\cite{mahabadi2018nonlinear,narayanan2019optimal,cherapanamjeri2022terminal}. 

Regarding top-$k$-NNs ordinal embeddings (Q3), in many important applications of embeddings, preserving distance information about nearby points is much more important than
preserving all distances. Indeed, it may be good enough to strictly maintain the order of the top-$k$ nearest points, and for far away objects to just label them as ``far''.   In such
scenarios it is natural to seek local embeddings that maintain only distances of close by neighbors. This has natural applications in ranking, search, and recommendations where often the few top results are viewed, and has been studied extensively in the metric embeddings literature, both in practice~(\cite{belkin2003laplacian,xiao2006duality}) and in theory under the name of \textit{local} (metric) embeddings~(\cite{abraham2007local,indyk2007nearest}) or local versions of dimension reduction~(\cite{schechtman2009lower}). 

Finally, regarding the behaviour of the above questions as we vary $k$ (Q4), we believe it is crucial to understand the various tradeoffs between the parameter $k$, the number of points $n$, and the target space dimension $d$, as various applications may need different value ranges for the parameters. 


\section{Preliminaries}
\label{sec:prelim}
Let $[n]$ denote the set $\{1,2,\cdots,n\}$. Let $X = ([n], \delta)$ be a metric space on $n$ points. Throughout our work, norm $\lVert\cdot\rVert$ is the standard Euclidean $\ell_2$ norm, unless otherwise noted.


\begin{definition}[\textit{Triplet} Ordinal Embedding, \textit{Triplet Dimension}]\label{def:triplet}  Let $\lVert\cdot\rVert$ be a norm on $\mathbf{R}^d$. We say that $\phi: X\to (\mathbf{R}^d,\lVert\cdot\rVert)$ is a \emph{triplet} ordinal embedding if
for every $x, y, z \in X$, we have the following: $\delta(x, y) < \delta(x, z) \iff \lVert\phi(x) - \phi(y)\rVert < \lVert\phi(x) - \phi(z)\rVert$. Moreover, the \emph{triplet dimension} of $X$ is defined to be the smallest
dimension $d$ of a triplet ordinal embedding of $X$.  We sometimes say the embedding is triplet-preserving, or that it preserves all triplet orders.
\end{definition}

\begin{definition}[\textit{Terminal} Ordinal Embedding] \label{def:terminal} Given $X = (V, \delta)$ with $|V|=n$ and a subset $T=\{t_1,\ldots, t_k\}\subseteq V$ of $k$ distinguished elements, which we call terminals, we say that $\phi: (X,T)\to (\mathbf{R}^d,\lVert\cdot\rVert)$ is a \emph{terminal} ordinal embedding if it preserves the distance orders from the rest of the points (the clients in our previous example) to the terminals (the servers), i.e., if
for every $t, t' \in T$ and $x,x'\in V\setminus T$, we have the following: $\delta(t, x) < \delta(t', x') \iff \lVert\phi(t) - \phi(x)\rVert < \lVert\phi(t') - \phi(x')\rVert$. The case $t=t'$ is also included.
\end{definition} 

\begin{definition}[\textit{Top-$k$-NNs} Ordinal Embedding]\label{def:top} Given $X$ as above, define $\mathrm{NN}(i)$ for $i\in X$ to be the set of $k\le n-1$ nearest neighbors of $i$ according to $\delta$. We say that $\phi: X\to (\mathbf{R}^d,\lVert\cdot\rVert)$ is a \emph{top-$k$-NNs} ordinal embedding if it preserves the distance orders among all elements in $NN(i),\forall i$, i.e., if
for every $i, j \in X$ and $i'\in \mathrm{NN}(i),j'\in \mathrm{NN}(j)$, we have the following: $\delta(i, i') < \delta(j, j') \iff \lVert\phi(i) - \phi(i')\rVert < \lVert\phi(j) - \phi(j')\rVert$. If in addition, we required $j\equiv i$ in the previous sentence, then we would get the problem of top-$k$-NNs ordinal embedding \emph{without mixed comparisons}.
\end{definition}

\section{Preserving Order on Triplets and the Triplet Dimension}
In this section, we prove the first main result about triplet-preserving ordinal embeddings whose goal is to respect all triplet comparisons for the distances in the original space. Then, we show a tradeoff for the dimension vs the relaxation, analogous to Bourgain's embedding theorem providing a tradeoff for the dimension vs the distortion. 
\subsection{Lower Bound for Exact Triplet Preservation}\label{subsec:exact}
\begin{theorem}\label{th:exact_triplet}
For every constant $\kappa>0$, and for every large enough $n$ (size of the dataset), no $d$-dimensional embedding in $\ell_2$ can be a triplet ordinal embedding (see Definition~\ref{def:triplet}), unless its dimension grows linearly as $d > \frac {n}{2+\kappa}$. (i.e., dimension must be roughly $n/2$ to preserve all triplets)
\end{theorem}

\begin{proof}
First, recall the superfactorial function $G(z)$ as defined on integers:

${\displaystyle G(n)={\begin{cases}0&{\text{if }}n=0,-1,-2,\dots \\\prod _{i=0}^{n-2}i!&{\text{if }}n=1,2,\dots \end{cases}}}$

This is Barne's special $G$-function which is related to the gamma function as $G(z+1)=\Gamma (z)\,G(z)$, with $G(1)=1$. Here we care about its asymptotic growth rate for integer $n$:
\begin{fact}\label{fact:triplets}
$G(n+1) = 0!1!\cdots(n-1)!$ and $\log G(n+1)={\frac {n^{2}}{2}}\log n+o\left({\frac {n^{2}}{2}\log n}\right) \approx$ $\frac {n^{2}}{2}\log n$. (The symbol $f(n)\approx g(n)$ simply means that the two quantities have asymptotically the exact same behavior with the exact same leading constant in front of the dominant term, where we can ignore lower order terms.)
\end{fact}
 Observe that the number of distinct triplet-orderings on $n\ge3$ items is at least as large as $G(n+1)$, since the first item can have any of the $(n-1)!$ permutations for its neighbors, the second item any of the $(n-2)!$ on its neighbors (excluding the first item) and so on. This is a good enough bound to prove the lower bound for the dimension as stated in the theorem.\footnote{We can compute the exact number of triplet orders, but the final lower bound for the dimension is asymptotically the same. The distinct triplet orders are $\tfrac{(n-1)!}{0!}\cdot \tfrac{n!}{2!}\cdot \tfrac{(n+1)!}{4!}\cdot \tfrac{(n+2)!}{6!} \cdots \tfrac{(2n-3)!}{(2n-4)!}$.}

\medskip
\noindent \textbf{Triplet-Preserving Embeddings as Low-degree Polynomials.} The second ingredient in the proof is to associate triplet-preserving embeddings with polynomials of degree $2$ and count how many different configurations could exist for resolving distance comparisons among $\delta(i,j)$ and $\delta(i,k)$. The embedding assigns $d$ coordinates to each of the $n$ elements so we can define the $n\times d$ matrix, say $X$, just by stacking row-by-row the coordinates of each of the $n$ elements; denote with $x_i$, the $i^{th}$ row of this matrix. We say that the matrix $X$ ``realizes the triplet order of $\delta$'' if the euclidean distances between pairs of rows $|\lvert x_i-x_j \rvert|$ and $|\lvert x_i-x_k \rvert|$ are consistent with $\delta$, i.e., if for every $i,j,k \in [n]$, we have:
\[
\delta(i,j)<\delta(i,k) \iff |\lvert x_i-x_j \rvert|<|\lvert x_i-x_k \rvert| 
\]
Notice that by squaring the euclidean distances, we can determine the inequality for a triplet $(i,j,k)$, by examining the sign of the following\footnote{With slight abuse of notation, we also use $X$ to denote the flattened version of the matrix, that is an $l$-dimensional variable with $l=n\cdot d$, and the resulting polynomials depend on (some) coordinates of $X$.} polynomial:
\[
p_{ijk}(X)\equiv p_{(i,j),(i,k)}(X)=||x_i-x_j||^2-||x_i-x_k||^2
\]
Notice how the sign of this polynomial determines the relative ordering for the distances among pair $(i,j)$ and pair $(i,k)$. In other words, there is a $1$-to-$1$ correspondence between the \textit{sign-patterns} of the polynomials $p_{ijk}(X)$ for $i,j,k\in [n]$ and the induced triplet orders.

\begin{definition}[Sign Patterns] Let $p_1,p_2,\ldots,p_m$ be real degree-$2$ polynomials over $l$ variables and let a point $u=(u_1,\ldots,u_l) \in \mathbb{R}^l$ be such that none of them vanishes. The \textit{sign-pattern} of the polynomials at point $u$ is the $m$-tuple $(\sigma_1,\ldots,\sigma_m)\in (-1,+1)^m$ where $\sigma_i=\mathrm{sign}(p_i(u))$. We use the notation $\mathrm{signs}(p_1,p_2,\ldots,p_m)$ to denote
the total number of different sign-patterns that can be obtained from $p_1,p_2,\ldots,p_m$ as point $u$ ranges over all points in $\mathbb{R}^l$.
\end{definition}
\begin{fact}[from~\cite{alon1985geometrical}]\label{fact:signs}
For any integer number $\beta$ between 1 and $m$, the total number of sign-patterns of  $m$ polynomials (as above) is upper bounded by:
\begin{equation}
\mathrm{signs}(p_1,p_2,\ldots,p_m)\le 4\beta\cdot(8\beta-1)^{l+\tfrac{m}{\beta}-1}\label{eq:signs}
\end{equation}
\end{fact}
\medskip
\textbf{Final Comparison: Orderings vs Signs.} The final step is a comparison. On the one hand, we already saw in Fact~\ref{fact:triplets} that $\log G(n+1)\approx \tfrac{n^2}{2}\log n$ and hence the number of distinct triplet orderings grows as $\tfrac{n^2}{2}\log n$. On the other hand, the number of polynomials is exactly $m=n\binom{n-1}{2}\approx n^3/2$ , because they are indexed by a triplet $(i,j,k)$ where distances between $(i,j)$ and $(i,k)$ are compared.

We set $n =c\cdot d$, for some constant $c$ (as we will see any $c>2$ suffices), and let the variables be $l = n \cdot d$ (every point gets assigned $d$ coordinates). We also set the parameter $\beta=\mu n=\mu cd$ for sufficiently large constant $\mu$. Taking the logarithm on both sides of~\eqref{eq:signs} (and ignoring the lower-order terms):
\[
\log(\mathrm{signs}(p_1,p_2,\ldots,p_m))\le (nd+\tfrac{n^3}{2\mu n})\log(8\mu n)
\]
The dominant term becomes $cd^2\log d$. Observe that for any $c>2$, this term is strictly smaller than  $\tfrac{n^2}{2}\log n$ and so for $n=cd \iff d=\tfrac{n}{c}<\tfrac n2$, there will be at least two distinct triplet-orderings that get mapped to the same sign pattern. This implies that any embedding that uses $d<n/2$ dimensions cannot be triplet-preserving (for worst-case instances), concluding the theorem. 
\end{proof}

\begin{remark}
Our results can be extended to other $\ell_p$ normed spaces too. For example, we can derive exactly the same lower bound for the dimension $d>\tfrac n2$ for any $\ell_p$ space (fixed $p\ge2$). The polynomials would have to be of degree $p$ instead of degree $2$;  their sign patterns are at most $2\beta p \cdot (4\beta p-1)^{l+\tfrac m\beta-1}$ ~(\cite{alon1985geometrical}), so $p$ affects only lower order terms.
\end{remark}

\subsection{Tradeoff for Triplet Dimension vs Relaxation}\label{subsec:non-exact}

\begin{theorem}\label{th:relax_triplet}
There is an absolute constant $c>0$ (taking $c=8$ suffices) such that for every integer $d$ and integer $n$, there is a metric space $T$ on $n$ points such that the \emph{triplet relaxation} of any ordinal embedding of $T$ into $d$-dimensional Euclidean space is at least $\tfrac{\log n}{\log d+\log\log n+c}-1$.
\end{theorem}
\begin{proof}
Let $G=(V,E)$ be a high-girth graph whose girth is $g=\tfrac{\log n}{\log d+\log\log n+8}$. We can assume that the number of edges $m=|E|\ge \tfrac14n^{1+\tfrac1g}>16nd\log n$, by known constructions~(\cite{sauer1970existence}). We will construct a large number $N$ of edge subgraphs of $G$, namely $G_1,G_2,\ldots,G_N$ (all have vertex set $V$), with the following crucial property:
\begin{equation}\label{eq:pairwise}
  \text{ For every pair } G_i,G_j, \exists v\in V: E_{G_i}(v)\setminus E_{G_j}(v)\neq \emptyset \text{ and } E_{G_j}(v)\setminus E_{G_i}(v)\neq \emptyset\tag{*}  
\end{equation}
Here $E_{G_i}(v)$ denotes the set of edges incident to $v$ within $G_i$. Property~\eqref{eq:pairwise} captures that from the perspective of vertex $v$, the subgraphs $G_i,G_j$ are significantly different.
\begin{definition}[Witness, ``Faraway'' Metric Pair]
For a pair of subgraphs $G_i,G_j$, we say that a vertex $v$ is a \emph{witness} (of faraway metric pair) $G_i,G_j$ if the following two relations hold: $E_{G_i}(v)\setminus E_{G_j}(v)\neq \emptyset \text{ and } E_{G_j}(v)\setminus E_{G_i}(v)\neq \emptyset$. The pair $G_i,G_j$ is said to be a \emph{faraway metric pair} if there exists such a witness vertex $v$.
\end{definition}

Restating property~\eqref{eq:pairwise}, we shall say that every pair $G_i,G_j$ is faraway, as such a pair differs significantly in their triplets from at least the perspective of the vertex $v$. Our next goal becomes how to generate lots of faraway metric pairs. We independently sample edges of $G$ with probability $\tfrac12$ and generate $N$ subgraphs $G_1,G_2,\ldots,G_N$. As we will see, we will pick $N=2^{bm}$ for $b<\tfrac12 \log_2\left(4/3\right)$.

\begin{lemma}
The set of subgraphs $G_1,G_2,\ldots,G_N$ generated as above satisfies property~\eqref{eq:pairwise} with high probability.
\end{lemma}
\begin{proof} Let us fix a pair $G_i,G_j$ and compute the probability it is not a faraway pair. Just for now, let us assume that the graph $G$ was a $k$-regular graph (every vertex had degree exactly $k$). This simplifies the exposition and we show how to drop this assumption next. Fix a vertex $v$. We have:
\[
\Pr[v \text{ is not witness for }G_i,G_j ]=\Pr\left[E_{G_i}(v)\setminus E_{G_j}(v)= \emptyset \text{ or } E_{G_j}(v)\setminus E_{G_i}(v)= \emptyset\right]
\]
Each of the two events has probability $\left(\tfrac{3}{4}\right)^k$ and to see why this is the case let's consider the first event $E_{G_i}(v)\setminus E_{G_j}(v)= \emptyset$ (the other event follows in the same way). Consider an edge $e=(v,w)$. If $e\in E_{G_i}(v)$ and $e\in E_{G_j}(v)$, or if $e\notin E_{G_i}(v)$, then $E_{G_i}(v)\setminus E_{G_j}(v)= \emptyset$. The latter events are disjoint and their union has probability $\tfrac12\cdot\tfrac12+\tfrac12=\tfrac34$, since each edge was independently included w.p. $\tfrac12$. This means that:
\[
\Pr[v \text{ is not witness for }G_i,G_j ]\le 2\cdot \left(\tfrac{3}{4}\right)^k
\]
At this point, notice that we would almost be done with the proof of the lemma, if we could somehow ensure that the event of whether or not a vertex $v$ is a witness for $G_i,G_j$ was independent from other vertices $v'$ being witnesses for $G_i,G_j$. This is because we could bound the probability $\Pr[\text{all } v \in V \text{ are not witnesses for }G_i,G_j]\le2^n\left(\tfrac{3}{4}\right)^{kn}$ and using a union bound over the $N^2$ pairs of subgraphs, we would get the lemma. Unfortunately, the independence of witnesses does not hold (as edges share endpoints), nor does the $k$-degree regularity assumption.

Let us now show how to circumvent the regularity and the independence assumptions. We will use a sequential process to determine whether a vertex $v$ is a witness or not, where we have to redefine the notion of a vertex neighborhood as we sequentially process the edges. Let us fix an ordering over the vertices of the graph $v_1,v_2,\ldots,v_n$. Let's redefine the neighborhoods of each vertex in an incremental way as follows:
\begin{itemize}
    \item First vertex $v_1$ gets all edges $E_G(v_1)$. Denote this set with $E_G'(v_1)$.
    \item Second vertex $v_2$ gets all edges in $E_G(v_2)$ except those already selected by $v_1$, i.e., all edges $E_G(v_2)\setminus E_G(v_1)$. Denote this set with $E_G'(v_2)$.
    \item Third vertex $v_3$ gets all edges $E_G(v_3)\setminus \left(E_G(v_1)\cup E_G(v_2)\right)$, denoted as $E_G'(v_3)$. We continue until the end, where the last vertex $v_n$ has no edges left.
\end{itemize}
Notice that each edge of $G$ contributes to exactly one neighborhood (the first one it is part of), so $\sum_{v\in V}|E_G'(v)|=m$, where $m$ is the total number of edges in $G$. With this new definition, we can indeed write:
\[
\Pr[\text{all } v \text{ are not witnesses for }G_i,G_j]\le 2\left(\tfrac34\right)^{|E_G'(v_1)|}\cdot 2\left(\tfrac34\right)^{|E_G'(v_2)|}\cdots
2\left(\tfrac34\right)^{|E_G'(v_n)|}\le2^n\left(\tfrac34\right)^{m}
\]
Finally, we take a union bound over the $N^2$ subgraphs to show that they satisfy property~\eqref{eq:pairwise} with high probability:
\[
\Pr[\text{there exists pair } G_i,G_j \text{ that is not faraway}]\le \binom{N}{2}2^n\left(\tfrac34\right)^{m}
\]
This bound goes to $0$, if we pick $N^2 \ll (\tfrac{4}{3})^m$ (e.g., $N=1.14^m$ suffices).  This concludes the proof of the lemma.
\end{proof}



In order to conclude the proof of the theorem, we still need to compare the number $N=1.14^m>1.14^{16nd\log n}>8^{nd\log n}$ of faraway metrics we generated, with the total number of sign patterns of polynomials given in Fact~\ref{fact:signs}. Again, we care about triplet distances so there are $n^3/2$ polynomials of degree $2$ over $dn$ variables, yielding a number of sign patterns $2^{(2+o(1))nd\log n}\ll 8^{nd\log n}$. 

By the pigeonhole principle, two distinct
metric spaces from our collection of $N$ subgraphs, get mapped to the same sign pattern so the distance orders in their embeddings are the same. Given that the graph has girth $g$ and using~\eqref{eq:pairwise}, this implies the relaxation in at least one of these
embeddings is at least $g-1$, completing the proof.
\end{proof}

\section{Terminal Ordinal Embeddings}\label{sec:terminal}
In this section, we present tight bounds for  \textit{terminal ordinal embeddings} (see Definition~\ref{def:terminal}), which is analogous to the notion of terminal  embeddings studied recently in metric embeddings~(\cite{elkin2017terminal,mahabadi2018nonlinear,narayanan2019optimal,cherapanamjeri2022terminal}). 

\subsection{Upper Bound for $k$ terminals}
Let $T=\{t_1,\ldots, t_k\}$ be the set of terminals, and $V\setminus T=\{v_1,\ldots, v_{n-k}\}$ be the set of the other points. For each pair $t \in T, v \in V$, let $r(t,v)$ be a unique integer in the range $1,\ldots,kn$ that specifies the rank of the distance between $t$ and $v$ among all $kn$ such distances. 

\begin{theorem}\label{th:upper_terminal}
Terminal ordinal embedding for $k$ terminals can be done with $k$ dimensions.
\end{theorem}

\begin{proof}
Our embedding (into a $k$-dimensional space) is as follows:
\begin{itemize}
    \item Each terminal $t_i \in T$ is mapped to $f(t_i)=-M e_i$, where $M$ is a ``large'' number to be specified soon, and $e_i$ has $1$ on the $i^{th}$ position and $0$ elsewhere.
    \item Each vertex $v$ is mapped into a $k$-dimensional point
 $f(v)=[ r(t_1,v), r(t_2,v),\ldots,r(t_k,v)]$.
\end{itemize}

\begin{lemma}
For large enough $M$ (picking $M=k^3 n^2$ suffices), our mapping preserves the distance orders between any pairs  $(t,v)$ and $(t',v')$.
\end{lemma} 
\begin{proof}
To see this, observe that
\[
||f(v)-f(t)||^2 = \sum_{s\neq t} r(s,v)^2 + \left(r(t,v)+M\right)^2 = \sum_s r(s,v)^2 + 2r(t,v)M + M^2
\]
Let $C(t,v)=\sum_s r(s,v)^2$, and observe that $C(t,v)\le k\cdot(kn)^2=k^3 n^2$. 

Thus, if we set $M=k^3 n^2$, then the term $2r(t,v)M+M^2$ ``dominates'', i.e., $||f(v)-f(t)||^2 < ||f(v')-f(t')||^2$ if and only if $r(t,v)<r(t',v')$.
\end{proof}
 So our embedding preserves the correct relationship between all pairs between $T$ and the rest of the points, and only uses $k$ dimensions.
\end{proof}

\begin{remark}
Note that our construction did not explicitly use that $k$ is fixed so it works for larger values of $k$ too. As we show next, for values of $k=o(n)$ this is tight, and for $k=\Theta(n)$ this is tight up to a constant of $2$ compared to the lower bound. Moreover, our approach works for other $\ell_p$ norms by choosing appropriate $M$.
\end{remark}

\subsection{Lower Bound for $k=\Theta(n)$ Terminals}

We first analyze the case where there are many items playing the role of terminals, specifically $k=\lambda n$ with $\lambda\in (0,1]$. We prove the following:

\begin{theorem}\label{th:exact_terminal}
For every constant $\lambda\in (0,1]$, and for every large enough $n$, no $d$-dimensional embedding in $\ell_2$ can be a terminal ordinal embedding for $k=\lambda n$ terminals, unless its dimension grows linearly as $d > \frac {n}{c}$, where $c>\tfrac{2}{2\lambda-\lambda^2}$. Equivalently, dimension $d>k(1-\tfrac{\lambda}{2})$ is needed for any terminal ordinal embedding on $k=\lambda n$ terminals. 
\end{theorem}
For $\lambda=1$, we recover as a special case a known result of Bilu and Linial for monotone maps~\cite{bilu2005monotone}.

\begin{proof}
We have $k$ terminals, so the total number of distinct distances to-be-preserved is:
\[
\text{Distinct Distance Elements} = (n-1) + (n-2) +\ldots + (n-k) \approx kn-\tfrac{k^2}{2}
\]
There are roughly $\left(kn-\tfrac{k^2}{2}\right)!$ different orderings, so taking logarithms we have that:
\begin{equation}\label{eq:term_order}
 \log (\#\text{orderings}) \approx \left(kn-\tfrac{k^2}{2}\right)\log\left(kn-\tfrac{k^2}{2}\right)
\end{equation}

On the other hand, the number of paired distances we can compare is:
\[
m= \text{Paired Terminal Distances}= \binom{kn-\tfrac{k^2}{2}}{2}\approx\tfrac{k^2n^2-k^3n+\tfrac{k^4}{4}}{2}
\]
For each paired terminal distance we associate a degree $2$ polynomial on at most $nd$ variables, the sign of which determines the outcome of the distance comparison. The logarithm of different sign patterns that can arise is bounded by:
\begin{equation}\label{eq:term_signs}
  \log(\mathrm{signs}(p_1,p_2,\ldots,p_m))\approx \left(nd+\frac{m}{\mu n^x}\right)\log n^x
\end{equation}
where $\mu$ can be chosen to be a large constant, and $\mu n^x\in [1,m]$. We can determine when the quantity in \eqref{eq:term_order} is asymptotically larger that the quantity in \eqref{eq:term_signs}: after setting $n=cd$ (for constant $c$), parameter $x=2$ and choosing $\mu$ large enough, we ultimately get (for sufficiently large $n$):
\[
\left(\lambda n^2-\frac{\lambda^2n^2}{2}\right)\log\left(\lambda n^2-\frac{\lambda^2n^2}{2}\right) \gg \left(cd^2+\frac{k^2n^2-k^3n+\tfrac{k^4}{2}}{2\mu n^2} \right)\log d^2
\]
The notation $\gg$ simply means that the dominant term in the left-hand side becomes larger than the dominant term in the right-hand side, for $n$ chosen sufficiently large. The above happens for:
\[
c^2d^2\left(\lambda-\frac{\lambda^2}{2}\right)2\log d > cd^2 2\log d \iff c > \frac{2}{2\lambda-\lambda^2}
\]
which finishes the proof of the theorem.
\end{proof}

\subsection{Lower Bound for $k=o(n)$ Terminals}\label{th:lower_terminal}

The difference from the previous subsection is that here the number of terminals is sublinear in $n$ (the total number of items). Interestingly, for technical reasons that will become clear very soon, there is a qualitative (and quantitative) difference compared to when $k=\Theta(n)$, that allows us to get tighter bounds compared to previously. Specifically, our lower bound for the sublinear case shows that if the dimension is $d< (1-\epsilon) k$ for any $\epsilon>0$, then the embedding will mess up some of the distance orders. Hence, this is a tighter bound compared to Theorem~\ref{th:exact_terminal}, where a term $\lambda/2$ needs to be subtracted.

\begin{theorem}\label{th:lower_sub}
Let $k=\sqrt{n}$. For any constant $c>1$ and dimension $d<\tfrac kc=\tfrac{\sqrt n}{c}$, no $d$-dimensional embedding can be a terminal ordinal embedding. Generally, this lower bound holds for any $k=\lambda n^{1-\beta}$ for $\lambda \in (0,1]$ and $\beta\in (0,1]$. Notice that $\beta$ is stricly larger than $0$, so $k$ is sublinear in $n$.
\end{theorem}
\begin{proof} The proof is found in the Appendix~\ref{app:proofs}. 
\end{proof}

\subsection{Lower Bound for Terminal Embeddings without Inter-terminal Comparisons}
Here we show tightness of our $k$ terminal ordinal embeddings, even if we ignore the inter-terminal distance pairs $(t,v)$ vs $(t',v'), t\neq t'$. 
\begin{theorem}\label{th:no-inter-terminal}
Even ignoring inter-terminal distance pairs (comparisons between $(t,v)$ and $(t',v'), t\neq t'$), for the case of $k=\lambda n$ the dimension needs to grow as $d>k(1-\tfrac{\lambda}{2})$ for any terminal ordinal embedding on $k=\lambda n$ terminals. For range of $k=\lambda n^{1-\beta}$ with $\lambda \in (0,1]$ and $\beta\in (0,1]$, no $d$-dimensional embedding for $d<\tfrac kc=\tfrac{\sqrt n}{c}$ can be a terminal ordinal embedding. 
\end{theorem}
\begin{proof} The proof is found in the Appendix~\ref{app:proofs}. 
\end{proof}

%


\begin{remark}
Our construction in Theorem~\ref{th:upper_terminal} is tight for $k=o(n)$: no less than $(1-\epsilon)k$ dimensions can perform $k$ terminal embeddings. For linearly many terminals $k=\lambda n$, there is still a small gap between upper and lower bound: upper is  that $d=k$ dimensions suffice, lower is that at least $d>k(1-\tfrac\lambda2)$ dimensions are necessary.
\end{remark}

\section{Top-$k$-NNs Ordinal Embeddings}\label{sec:top-k}
Let $\mathrm{NN}(i)$ denote the set of $k$ nearest neighbors to point $i$ according to the original distance $\delta$. The top-$k$-NN ordinal embedding (see Definition~\ref{def:top}) asks to preserve the ordering on the set of distances for $\mathrm{NN}(i)$, for every $i$.

\begin{theorem}\label{th:topNNs}
Any top-$k$-NNs ordinal embedding needs to have a dimension that grows as $d=\Omega\left(k \right)$. Without mixed comparisons (Definition~\ref{def:top}), the bound becomes $d=\Omega\left(\tfrac{k\log k}{\log n} \right)$.
\end{theorem}
\begin{proof}
The proof is found in the Appendix~\ref{app:proofs}. 
\end{proof}

\section{Conclusion}
We studied several basic questions that arise in the context of ordinal embeddings (also contrastive embeddings) related to nearest-neighbor, ranking, recommendations, crowdsourcing and psychometrics, where we want to find euclidean representations that only respect the order among measured distances between $n$ items, rather than the lengths of those distances. For the well-motivated problem of triplets where information of the form ``is $i$ closer to $j$ or to $k$?'' is provided, we give almost tight lower bounds for the necessary dimension so as to preserve the triplets relations either exactly or approximately. Going beyond triplets, we study the interesting scenario of terminal ordinal embeddings and we present matching upper and lower bounds. Finally, we present lower bounds for the top-$k$-nearest-neighbors ordinal embeddings problem.

\paragraph{Acknowledgement:} This research was supported in part by the
NSF TRIPODS program (award DMS-2022448), Simons
Investigator Award and MIT-IBM Watson AI Lab.
\bibliographystyle{abbrvnat}
\bibliography{references.bib}
\appendix
\section{Omitted Proofs}\label{app:proofs}

\begin{proof}[Proof of Theorem~\ref{th:lower_sub}]
The proof proceeds in analogous manner as before, but the difference is that the term $kn-\tfrac{k^2}{2} \approx kn$, i.e., $\tfrac{k^2}{2} \ll kn$ for large $n$ and so it won't affect the computations. We present the relevant calculations below. The important quantities we need to compare are:
\begin{itemize}
    \item The logarithm of the total number of orderings:  $\left(kn-\tfrac{k^2}{2}\right)\log\left(kn-\tfrac{k^2}{2}\right)  \text{ vs }$ 
    \item The logarithm of sign patterns of polynomials: $\left(nd+\tfrac{m}{\mu n^x}\right)\log n^x$, with parameters $\mu,x$: $\mu n^x \le m$, where $m=\tfrac{k^2n^2-k^3n+\tfrac{k^4}{4}}{2}$, same as before.
\end{itemize}
We can make several simplifications based on the fact that $k=\lambda n^{1-\beta}=o(n)$. Ignoring lower order terms we have that:
\begin{itemize}
    \item  $k^2=o(kn)$ and also $m\approx \tfrac{k^2n^2}{2}=\Theta(k^2n^2)$.
    \item We can pick $x$ such that $n^x\approx kn=o(n^2)$ and $\mu n^x\ll m$ for any constant $\mu$.
    \item Let $k=c\cdot d$ for some constant $c$ (to be determined below).
\end{itemize}
Taking these into account we need to compare:
$cdn\log kn$ and $\left(nd+\tfrac{cdn}{2\mu}\right)\log kn$. 

We can set $\mu$ large enough so that the former quantity is larger than the latter whenever $c>1$, i.e., $k\ge cd$ for constant $c>1$. 
\end{proof}

\begin{proof}[Proof of Theorem~\ref{th:no-inter-terminal}]
The proof follows almost the same calculations as in Theorem~\ref{th:exact_terminal} and Theorem~\ref{th:lower_sub}. The reason is the following: if we cared about only same-terminal distance pairs we get:
\begin{itemize}
    \item  $\#\mathrm{orders} > (n-1)! (n-2)! ... (n-k)!$
    \item Again the number of paired comparisons is at most $m\approx \tfrac{k^2n^2}{2}=\Theta(n^2)$, so this is also the number of corresponding polynomials that determine the orderings.
    \item Number of orders is larger than $[(n-k)!]^k \implies \log\#\mathrm{orders} > k \log[(n-k)!]$. 
    \item The leading term again is $kn\log n$, and this is enough to beat the number of polynomials and get the tradeoffs we want. If we wanted to be more precise, we can use that $\log[(n-1)! (n-2)! ... (n-k)!]\approx (kn-\tfrac{k^2}{2})\log(n-k)$.
\end{itemize}
We can conclude that the same lower bounds (with same constants) hold.
\end{proof}

\begin{proof}[Proof of Theorem~\ref{th:topNNs}]
Following the approach as in the previous proofs, we compare the following:
\begin{itemize}
    \item Allowing for mixed comparisons, the number of possible orderings among distances in $\cup_{i}\mathrm{NN}(i)$ for all the items $i$ can be close to $(nk -o(nk))!$. Hence the logarithm is $\approx nk\log n$.

    \item Allowing for mixed comparisons, there are at most $m=\binom{nk}{2}$ paired distance comparisons (if we don't allow for mixed comparisons, $m$ is even smaller). So the logarithm of the sign patterns for the corresponding polynomials is $(nd+\tfrac{m}{\mu n^x})\log(\mu n^x)$ for any $x,\mu:$ $\mu n^x\le m$. For this calculation, it's sufficient to pick $x$: $n^x\approx kn=O(n^2)$ so $x \le 2$.
    \item Comparing the logarithms we get that there are more orderings than sign patterns as long as: $nk\log n \gg ndx\log n$. So the dimension $d$ must be at least $\Omega(k)$.
    \item Without mixed comparisons, the number of orderings per point is $k!$ since $|\mathrm{NN}(i)|=k$. So we get in total at least $(k!)^{\Theta(n)}$  distance orderings (so $\log(k!)^{\Theta(n)} \approx nk\log k$). This implies that for the embedding to preserve distance orderings it has to use dimension $d=\Omega\left( \tfrac{k\log k}{\log n}\right)$. 
\end{itemize}
\end{proof}

\end{document}